\documentclass[12pt]{article}
\usepackage{amsmath,amsfonts,amsthm,amssymb}

\usepackage{fancybox}
\usepackage{amstext}

\newcommand{\version}{December 13, 2015}

\renewcommand\footnotemark{}

\swapnumbers
\font\notefont=cmsl8

\theoremstyle{plain}
\newtheorem{thm}{THEOREM}[section]
\newtheorem{lm}[thm]{LEMMA}
\newtheorem{cl}[thm]{COROLLARY}

\theoremstyle{definition}

\theoremstyle{remark}

\newcommand{\upchi}{\raise1pt\hbox{$\chi$}}
\newcommand{\R}{{\mathord{\mathbb R}}}
\newcommand{\C}{{\mathord{\mathbb C}}}
\renewcommand{\P}{{\mathord{\mathbb P}}}

\newcommand{\N}{{\mathord{\mathbb N}}}
\newcommand{\rS}{{\mathord{\mathbb S}}}

\newcommand{\cT}{{\mathord{\cal T}}}

\newcommand{\cH}{{\mathord{\cal H}}}

\newcommand{\cW}{{\mathord{\mathcal W}}}
\newcommand{\Tr}{{\mathord{\rm Tr}}}

\newcommand{\su}{{$SU(N)\ $}}

\usepackage{color}\definecolor{darkred}{cmyk}{0,1,1,0.4}
\newcommand{\ch}[1]{{#1}}
\begin{document}
\title{Proof of the Wehrl-type Entropy Conjecture for Symmmetric
  $SU(N)$ Coherent States \relax\footnotetext{\copyright 2015 \ by the
    authors. This article may be reproduced in its entirety for
    non-commercial purposes.}  \footnote{Work partially supported by
    NSF grant PHY-1265118 (EHL) and ERC grant ERC-2012-AdG 321029
    (JPS). The authors also thank IPAM at UCLA, IHP in Paris, and the
    Erwin Schr\"odinger institute in Vienna, where part of this work
    was done.}  }

\baselineskip=20pt

\author{
  \begin{tabular}{ccc} 
    Elliott H. Lieb&\hspace{2cm}& Jan Philip Solovej \\     \\
    \normalsize Departments of Physics and Mathematics && 
    \normalsize Department of Mathematics\\ 
    \normalsize Jadwin Hall, Princeton University && \normalsize University
of Copenhagen\\
    \normalsize Washington Road&&\normalsize Universitetsparken 5\\
    \normalsize Princeton, N.J. 08544  & &
    \normalsize DK-2100 Copenhagen, Denmark\\
    \normalsize {\it e-mail\/}: lieb@princeton.edu &&
    \normalsize {\it e-mail\/}: solovej@math.ku.dk
  \end{tabular}
  \bigskip
  \date{\version}}

                                \maketitle

\begin{abstract}
  The Wehrl entropy conjecture for coherent (highest weight) states in
  representations of the Heisenberg group, which was proved in 1978
  and recently extended by us to the group $SU(2)$, is further
  extended here to symmetric representations of the groups $SU(N)$ for
  all $N$. This result gives further evidence for our conjecture that
  highest weight states minimize group integrals of certain concave
  functions for a large class of Lie groups and their representations.
\end{abstract}

 \section{Introduction}  \label{intro} 
 
 With the aid of coherent states, A. Wehrl \cite{AW} introduced the
 idea of a `classical entropy' associated to a quantum density
 matrix. He showed that it has the desirable feature of being positive
 and conjectured that the minimum entropy over all density matrices
 would be achieved by a one-dimensional projector onto a coherent
 state. He stated the conjecture only for Glauber coherent states on
 $L^2(\R^n)$, and this was proved shortly thereafter in \cite{L2}, in
 which the conjecture was extended to $SU(2)$. This $SU(2)$ conjecture
 was finally settled by us 35 years later \cite{LS}, although there
 were several special cases proved earlier \cite{Bo,GZ,Lu,Sch,Scu,Su}.

 We believe that the analog of Wehrl's conjecture should hold true, at
 least for a wide class of Lie groups (see also \cite{BZ,SZ}). If so,
 this would presumably have some general significance for
 representation theory since there are not very many theorems about
 integrals over the group of finite dimensional representations.  The
 progress reported here is a proof of the conjecture for all symmetric
 representations of $SU(N)$, i.e., the representations corresponding
 to one-row Young diagrams. In the case of $SU(2)$ there are no other
 representations. For $SU(N)$ our proof trivially generalizes to the
 conjugate of the symmetric representations which are unitarily
 equivalent to the representations with young diagrams having $N-1$
 rows of equal length. We note that the proof we give now is much
 improved over that in \cite{LS}, which did not generalize to $SU(N)$
 for $N\geq3$.  The improvement utilizes a combinatorial identity (see
 (\ref{eq:chiribella}) below) that replaces some $SU(2)$-specific
 relations used before.  Non-symmetric representations, corresponding
 to other multi-rowed Young diagrams, are not addressed.

 As in \cite{LS} we extend the conjecture to all concave functions,
 not only to the entropy function $-x\ln x$. Again, we prove the
 conjecture for $SU(N)$ Wehrl entropy as the infinite dimensional
 limit of a sequence of finite dimensional theorems.

We begin with a very brief reminder of the conjecture. Let $\cH$ be a
Hilbert space for an irreducible representation of a suitable
(e.g., compact, simple and connected) Lie group $G$,
let $\Omega_I\in \cH $ be a normalized highest weight vector, called a
{\bf coherent state vector}, let 
$\Omega_R = R\Omega_I $ for $R\in G$,  and let $v\in \cH $ 
be any other normalized vector. Naturally, $\Omega_R$ is also a highest
weight vector.  Next form the Husimi function \cite{Hu} on $G$ , which is defined by the inner
product $|\langle v\, |\, \Omega_R \rangle|^2$. With $f(x) = -x\ln x$, we 
define Wehrl's classical entropy $S(v)$ by the (normalized)
Haar measure integral
\begin{equation}
\label{eq:int}
 S(v) = \int_G f(|\langle v\, | \, \Omega_R\rangle|^2)\,  {\rm d} R.
\end{equation}

The integral above may be considered on the space of equivalence
classes corresponding to $R\sim R'$ if $\Omega_R$ and $\Omega_{R'}$
are equal up to a phase. This space of equivalence classes (the
co-adjoint orbit of highest weight vectors) has a natural symplectic
structure for which the Liouville measure agrees with the measure
inherited from the Haar measure (see, e.g., \cite{Si}).  In this sense the Husimi function
becomes a probability distribution on a classical phase space, and the
number $S(v)$ is a classical entropy on that space, corresponding to the state $v$. The
symplectic structure, however, plays no role for our purpose, and we will
consider $S(v)$ as an integral over the group, as defined in
(\ref{eq:int}).

The conjecture is that $S(v)$ is minimized when the normalized $v$ is
any of the vectors $\Omega_R$. Note that $S(v)>0$ since the Husimi
function is less than one almost everywhere, unlike the Boltzmann
entropy which can be negative.  The extended conjecture is that this
minimization property also holds if $f$ in (\ref{eq:int}) is any
concave function.

The group considered in this paper is $SU(N)$ and the irreducible
representations are the totally symmetric ones (and their conjugates), 
defined in the next section.
 
The Husimi function that associates a classical distribution function
to a quantum state can be generalized to certain maps associating
quantum states on one representation space to states on another
representation space.  The maps in question are completely positive
trace preserving and are called `quantum channels'.  As we shall
discuss below the particular channels we study are sometimes referred
to as the `universal quantum cloning channels'.  A further
generalization of the entropy conjecture is that it has a natural
extension to these channels, i.e., that highest weight vectors
minimize the trace of concave functions of the channel output (see
Theorem~\ref{thm:main}). In particular, we thus determine the minimal
output entropy of the universal cloning channels. In the proof we use
that the minimal output of cloning channels agrees with the
minimal output of what is called the `measure-and-prepare channels'
(see \cite{Ch}). As a corollary we therefore also determine their
minimal output entropy (See Theorem~\ref{cl:Wk}).

The original Wehrl-type conjecture for the Husimi function will be
derived as the limit of the finite dimensional results.  In a similar
way, we showed in \cite{LS} how to prove the original Glauber and
$SU(2)$ (Bloch) coherent state conjectures from the infinite
dimensional limit of finite dimensional representations of
$SU(2)$. Although the story begins with the proof in \cite{L2} of the
Wehrl entropy conjecture for Glauber states, it is only in \cite{LS}
that the generalization to all concave functions was achieved for
Glauber states.
 
\section {Symmetric Irreducible Representations of $SU(N)$}

The symmetric irreducible representations (irrep) of $SU(N)$ are
obtained by taking $M$ symmetric copies of the fundamental
representation. We consider totally symmetric tensor products
of $N$-dimensional complex space, i.e., for $M\in\N$,
$\cH_M=P_M\bigotimes^M\C^N$.  Here,
$P_M:\bigotimes^M\C^N\to\bigotimes^M\C^N$ is the projection onto the
symmetric subspace, i.e., for $u_i \in \C^N$,
\begin{equation} \label{eq:pm}
P_Mu_1\otimes\cdots\otimes u_M=\frac1{M!}\sum_{\sigma\in
S_M}u_{\sigma(1)}\otimes\cdots\otimes u_{\sigma(M)}.
\end{equation}

The group \su acts on $\bigotimes^M\C^N$ equally on
each factor, i.e., $R\in SU(N)$ acts as 
$R\otimes\cdots\otimes R$. This action commutes with $P_M$ and hence acts
on $\cH_M$. It is well-known that this is an irreducible representation, as
explained in the appendix.

The highest weight vectors are the ones in which all the $u_i$ are the
same normalized vector $u$, i.e., a highest weight vector is of the form
$\otimes^{M} u $, and projectors onto such vectors are called coherent states.

Recall that a density matrix on a Hilbert space is a positive
semi-definite operator of unit trace.
Our notation here is that $\langle u|v\rangle$ is the inner product of vectors $u$ and $v$, while 
$\langle u|\rho|v\rangle$ is the inner product of $u$ with $\rho v$. The projector onto a normalized vector 
$|u\rangle$ is denoted $|u\rangle\langle u|$. 
\begin{thm}[{\bf Generalized Wehrl Inequality}]\label{thm:wehrl}
Let $f:[0,1]\to\R$ be a concave function. Then for any
density matrix $\rho$ on $\cH_M$ we have
\begin{equation}\label{eq:conjecture}
\int_{SU(N)}f(\langle\otimes^M(Ru)|\rho|\otimes^M(Ru)\rangle) {\rm d} R
\geq\int_{SU(N)}
f(|\langle\otimes^M(Ru)|\otimes^Mu\rangle|^2) {\rm d} R .
\end{equation}
Here $u$ is any (by $SU(N)$ invariance) normalized vector in
$\C^N$. In other words the integral on the left is minimized for
$\rho=|\otimes^Mu\rangle\langle\otimes^Mu|$, i.e., $\rho$ is a coherent state.
\end{thm}
This theorem is proved in Section~\ref{sec:semiclassics}. 

The classical phase space discussed in the introduction is, in this
case of the symmetric irreps of $SU(N)$, the space of pure quantum states
on the one-body space $\C^N$. This is the space of unit vectors in
$\C^N$ modulo a phase, i.e., the complex projective space
$$
\C\P^{N-1}=\{u\in\C^N\ |\ |u|=1\}/\sim,
$$
where two vectors are equivalent under $\sim$ if they agree up to
multiplication by a complex phase. The complex projective space
$\C\P^{N-1}$ is a classical phase space, i.e., a symplectic
manifold\footnote{As already stated, we will not be concerned with the
  symplectic 2-form on $\C\P^{N-1}$, known as the Fubini-Study form.
  We will only need the corresponding volume form, which corresponds
  to the normalized Haar measure on $SU(N)$.}. For any function $h$
defined on $\C\P^{N-1}$ we have the equivalence of the normalized
integrations
\begin{equation}\label{eq:su-int}
\int_{SU(N)} h(Ru) {\rm d} R=  \int_{\{u\in\C^N : |u|=1\}} h(u) {\rm d} u
=\int_{\C\P^{N-1}} h(u) {\rm d} u,
\end{equation}
where we have abused notation and identified unit vectors $u$ with their
equivalence class in $\C\P^{N-1}$. The first integration above is over the
$N(N+1)/2$ dimensional real manifold $SU(N)$, the middle integration is
over the $2N-1$ dimensional real sphere, and the last integral is over the
$2N-2$ dimensional real manifold $\C\P^{N-1}$.
In the special case $N=2$ we have
that $\C\P^1$
is the 2-sphere $\rS^2$ (the Bloch sphere). For $N\geq 3$, the compact
manifold $\C\P^{N-1}$ is not a sphere.

\section{The Quantum Channels}
In order to prove the main theorem we introduce maps from operators
on the space $\cH_M$ to operators on $\cH_{M+k}$ for some $k\geq0$.

If $\Gamma$ is an operator on the symmetric tensor
product $\cH_M$ we can extend $\Gamma$ uniquely to an operator which we also call $\Gamma$ on the
larger space $\bigotimes^M\C^N$ such
that 
$$
\Gamma=P_M\Gamma=\Gamma P_M =P_M \Gamma P_M.
$$
That is, $\Gamma =0$ on vectors of non-symmetric symmetry type.

We can write $\Gamma$ using second quantization as 
$$
\Gamma=\frac1{M!}\sum_{i_1=1}^N\cdots\sum_{i_M=1}^N\sum_{j_1=1}^N\cdots\sum_
{j_M=1}^N
\Gamma_{i_1,\ldots,i_M;j_1,\ldots,j_M}a^*_{i_1}\cdots
a^*_{i_M}a_{j_M}\cdots a_{j_1},
$$ 
where we have introduced the matrix elements 
for the extended $\Gamma$ as
$$
\Gamma_{i_1,\ldots,i_M;j_1,\ldots,j_M}=\langle u_{i_1}\otimes\cdots\otimes
u_{i_M}|\Gamma|u_{j_M}\otimes\cdots\otimes u_{j_1}\rangle
$$
using an orthonormal basis $\{u_i\}_{i=1}^N$ for $\C^N$ and the
corresponding Fock space creation and annihilation operators
$a^*_i=a^*(u_i)$, $a_i=a(u_i)$ satisfying the well known canonical
commutation relations $[a_i,a_j]=0= [a_i^*,a_j^*]$ and $[a_i, a_j^*] =
\delta_{i,j}$. A more extensive discussion of creation and
annihilation operators can be found in Section 4 of our earlier $SU(2)$
paper \cite{LS}.

We study the operator 
$$
\cT^k(\Gamma)=\sum_{i_1,\ldots,i_k}a^*_{i_1}\cdots a^*_{i_k} \Gamma
a_{i_k}\cdots a_{i_1}
$$
on $\cH_{M+k}$.  The map $\cT^k$, which maps operators on $\cH_M$ to
operators on $\cH_{M+k}$, is completely positive. Note that we have
not normalized it to be trace preserving, in fact,
$$
\Tr_{M+k} \cT^k(\Gamma)=\frac{(M+k+N-1)!}{(M+N-1)!}\Tr_M\Gamma,
$$
where $\Tr_M$ refers to the trace in $\cH_M$. Thus 
$$
\widehat\cT^k=\frac{(M+N-1)!}{(M+k+N-1)!}\cT^k
$$
is completely positive and trace preserving.
We may also write the map $\cT^k$ as 
\begin{equation}\label{eq:cloning}
\cT^k(\Gamma)=\frac{(M+k)!}{M!}P_{M+k}\left(\left(\bigotimes^kI_{\C^N}
\right) \otimes\Gamma\right) P_{M+k}.
\end{equation}
In this form we recognize the map $\widehat\cT^k$ as the universal
$M$-to-$M+k$ cloning channel \cite{Br,Ch,GM,We}. The No-Cloning
Theorem states that exact cloning of a quantum state is
impossible. The universal cloning channels achieve the best degree of
cloning for general input states. We thank Kamil Br\'adler for pointing
out the relation to cloning channels.

\section{The Main Theorem for Quantum Channels}
Our main result on the cloning channels $\widehat\cT_k$ is that
coherent states minimize the trace of concave functions of the channel
output. If the concave function is $f(x)=-x\ln(x)$ this theorem says
that coherent states, which are pure states, give the minimal output
von Neumann entropy. Since we can prove the optimality of coherent
states for all concave functions we refer to this as the generalized
minimal output entropy.
\begin{thm}[{\bf Main Theorem: Generalized minimal output entropy of
    $\cT_k$}] \label{thm:main} For any concave function $f:[0,1]\to\R$
  and any density matrix $\rho$ on $\cH_M$ we have
$$
\Tr_{M+k}f(\widehat\cT^k(\rho))\geq \Tr_{M+k}f(\widehat\cT^k(|\otimes^Mu\rangle\langle \otimes^Mu|)).
$$
For $M\geq 2$ and $k>0$ and $f$ strictly concave equality holds if and only if
$\rho=|\otimes^Mu\rangle\langle \otimes^Mu|$ for some unit vector
$u\in\C^N$.
\end{thm}
We will prove this by establishing that the sequence of ordered
eigenvalues of $\widehat\cT^k(|\otimes^Mu\rangle\langle \otimes^Mu|)$
majorizes the ordered sequence of eigenvalues of $\widehat\cT^k(\rho)$
for any density matrix $\rho$. The fact that majorizing eigenvalues is
equivalent to minimizing traces of concave functions is not difficult to prove and is known as
Karamata's Theorem (see \cite{LSe}). We used it also in our earlier paper \cite{LS} for $SU(2)$.

Recall that one ordered sequence of numbers $x_1\geq x_2
\geq\cdots\geq x_\mu$ majorizes another, $y_1\geq y_2 \geq\cdots \geq
y_\mu$, if, for each $1\leq k \leq \mu$, $\sum_{j=1}^k x_j \geq
\sum_{j=1}^k y_j$, and with equality for $k=\mu$. If $X$ and $Y$ are
Hermitian matrices we write $X\succ Y$ if the ordered eigenvalue
sequence of $X$ majorizes the ordered eigenvalue sequence of $Y$.  Our
main theorem above is thus a consequence of the following
majorization theorem.

\begin{thm}[{\bf Coherent States Majorize}]\label{thm:majorization} 
  Let $\Gamma$ be a positive semi-definite operator on $\cH_M$. The
  ordered sequence of eigenvalues of $\cT^k(\Gamma)$ is majorized by
  the ordered sequence of eigenvalues of
  $\Tr_M(\Gamma)\cT^k(|\otimes^Mu\rangle\langle \otimes^Mu|)$ for any
  unit vector $u\in \C^N$. Moreover, for $M\geq2$ and $k>0$
  strict majorization holds unless
  $\Gamma=\Tr_M(\Gamma)|\otimes^Mu\rangle\langle \otimes^Mu|$. The norm
  $\Tr_M(\Gamma)$ occurs since we have not assumed $\Gamma$ to have
  unit trace.
\end{thm} 
\begin{proof}
  We will use induction on $k$. The case $k=0$ is trivial as
  $\cT^0(\Gamma)=\Gamma$ and any positive semidefinite operator is
  majorized by a rank 1 operator with the same trace.

  More generally, we may assume that $\Gamma$ is a rank one
  projection. The assumption that $\Tr_M(\Gamma)=1$ is, trivially, no
  loss of generality. That the rank may be assumed to be one follows
  from the fact that if $A$ and $B$ are both majorized by $C$ then for
  $0\leq \lambda\leq 1$ we have that $\lambda A+(1-\lambda)B$ is
  majorized by $C$. Alternatively, we have that the partial sum of
  eigenvalues is a convex function.  Now simply write the spectral
  decomposition of $\Gamma$, i.e.,
  $\Gamma=\sum_{p=1}^r\lambda_p|\psi_p\rangle\langle \psi_p|$, where
  $0\leq \lambda_p\leq 1$ with $\sum_{p=1}^r\lambda_p=1$. If we can
  show that each $\cT^k(|\psi_p\rangle\langle \psi_p|)$ is majorized
  as claimed then the result follows for $\Gamma$.

  We will also use the following well known, simple observation. 
  \begin{lm} If $v_1,\ldots,v_m$ are vectors in a Hilbert space $\cH$
    then the operator $\sum_{i=1}^m|v_i\rangle\langle v_i|$ has the
    same non-zero eigenvalues as the $m\times m$ Gram matrix with entries
$\langle v_i| v_j\rangle$.
  \end{lm}
\begin{proof} Let $A:\C^m\to \cH$ be the linear map
$$
(z_1,\ldots,z_m)\mapsto z_1|v_1\rangle+\ldots+z_m|v_m\rangle.
$$
Its adjoint is the map $A^*:\cH\to \C^m$ given by 
$$
A^*|v\rangle=(\langle v_1|v\rangle,\ldots, \langle v_m|v\rangle).
$$
Then $AA^*=\sum_{i=1}^m|v_i\rangle\langle v_i|$ is an operator from
$\cH$ to itself and $A^*A$ is the linear map on $\C^m$ corresponding
to the Gram matrix. The non-zero eigenvalues of $AA^*$ are always the same
as those of $A^*A$, for any $A$.
\end{proof}
We assume now that the main theorem has been proved for all values
$0,1,2,...,k-1$. As explained above we may assume that
$\Gamma=|\psi\rangle\langle\psi|$, where $\psi$ is a unit vector in
$\cH_M$. According to the lemma $\cT^k(|\psi\rangle\langle\psi|)$ has the
same
non-zero eigenvalues as the matrix 
$$
W^\psi_{j_1,\dots, j_k;\, i_1,\dots,i_k} =
\langle\psi|a_{i_1}\cdots a_{i_k}a^*_{j_k}\cdots a^*_{j_1}|\psi\rangle.
$$
This matrix represents the operator on the space
$\cH_k$, given in second quantization by
\begin{equation}\label{eq:Wdef}
  \cW_k(|\psi\rangle\langle\psi|)=\frac1{k!}\sum_{i_1=1}^N\cdots\sum_{i_k=1}^N\sum_{j_1=1}^N\cdots\sum_
  { j_k=1}^N
  \langle\psi|a_{i_1}\cdots a_{i_k}a^*_{j_k}\cdots a^*_{j_1}|\psi\rangle
  a^*_{i_1}\cdots
  a^*_{i_k}a_{j_k}\cdots a_{j_1}.
\end{equation}
The map $\cW_k$ from operators on $\cH_M$ to operators on $\cH_k$ is again a completely positive map, i.e., 
if we normailze it to be trace preserving the resulting map $\widehat\cW_k$ would be a quantum channel.
We thus want to prove that 
\begin{equation}\label{eq:Wmajorization}
\cW_k(|\otimes^Mu\rangle\langle\otimes^Mu|)\succ \cW_k(|\psi\rangle\langle\psi|).
\end{equation}

By normal ordering the creation and annihilation operators (which
means utilizing the commutation relations to switch the creation
operators to the left of the annihilation operators), inside the
expectation value, we can express $\cW_k(|\psi\rangle\langle\psi|)$ in terms of the reduced
density matrices. These are the operators $\gamma^\ell_\psi$,
$\ell=0,\ldots$ defined on $\cH_\ell$ by
$$
\gamma^\ell_\psi=\frac1{\ell!}\sum_{i_1=1}^N\cdots\sum_{i_\ell=1}^N\sum_{
j_1=1}^N\cdots\sum_{j_\ell=1}^N
\langle\psi|a^*_{i_1}\dots a^*_{i_\ell}a_{j_\ell}\cdots
a_{j_1}|\psi\rangle a^*_{j_\ell}\cdots a^*_{j_1}a_{i_1}\dots a_{i_\ell},
$$
with the normalization convention $\Tr
\gamma^{(k-\ell)}_\psi=\frac{M!}{(M-k+\ell)!}$ (they vanish if
$k-\ell>M$).  In fact, as we commute creation operators $a_i^*$ to
the left of annihilation operators $a_j$, we will create delta
functions $\delta_{ij}$ with positive coefficients, and it is thus
evident that there will be positive constants $C_\ell$,
$\ell=0,\ldots,k$ (the exact values are not important to us) such that
\begin{equation}\label{eq:chiribella}
\cW_k(|\psi\rangle\langle\psi|)= \sum_{\ell=0}^kC_\ell \sum_{i_1=1}^N \cdots\sum_{i_\ell=1}^N
a^*_{i_1}\cdots
a^*_{i_\ell}\gamma^{(k-\ell)}_\psi a_{i_\ell}\cdots
a_{i_1}=\sum_{\ell=0}^kC_\ell \cT^\ell(\gamma^{(k-\ell)}_\psi ).
\end{equation}
The explicit constants in this formula were derived in \cite{Ch} where
$\widehat\cW_k$ was called the `universal
measure-and-prepare' channel.
See also \cite{LNR} and \cite{Ha} (Theorem 7) for alternative
calculations of the constants, which we recall are not important for
our application of the formula.

From the induction hypothesis we see that for all $\ell\leq k-1$
$$
\cT^\ell(\gamma^{(k-\ell)}_{\otimes^Mu} )\succ
\cT^\ell(\gamma^{(k-\ell)}_\psi ).
$$
For $\ell=k$, however, this is trivial since $\gamma^{0}_{\otimes^Mu} = \gamma^{0}_\psi =1$ and thus 
$$
\cT^k(\gamma^{0}_{\otimes^Mu} )=\cT^k(\gamma^{0}_\psi )=\sum_{i_1=1}^N
\cdots\sum_{i_k=1}^N a^*_{i_1}\cdots
a^*_{i_k}a_{i_k}\cdots
a_{i_1}=k!I_{\cH_k}.
$$
We now use the following simple observation. 

\begin{lm}\label{lm:majorsums} Consider Hermitean operators $A_1,\ldots,A_k$ that can
be  diagonalized in the same basis and such that
the eigenvalues are simultaneously ordered decreasingly. In other
words the operators are all non-decreasing functions of the same
operator, i.e.,  $A_i=f_i(A)$, $i=1,2,\ldots,k$ where $A$ is Hermitean
and $f_1,f_2,\ldots, f_k:\R\to\R$ are non-decreasing. 
If the Hermitean operators $B_1,\ldots,B_k$ satisfy 
$A_1\succ B_1,\ldots,A_k\succ B_k$ then 
$$
A_1+\ldots+A_k\succ B_1+\ldots+ B_k.
$$
\end{lm}
\begin{proof} To see this assume that $u_1,\ldots, u_q$ are the first
  $q$ eigenvectors in the basis diagonalizing $A$ and hence
  $A_1,\ldots,A_k$. If $v_1,\ldots,v_q$ are orthonormal eigenvectors
  for $B_1+\ldots+ B_k$ corresponding to the top $q$ eigenvalues
  $\mu_1,\ldots,\mu_q$ then (by the min-max principle)
\begin{eqnarray*}
\mu_1+\ldots+\mu_q=\sum_{j=1}^q\langle v_j|B_1+\ldots+
B_k|v_j\rangle&=&\sum_{j=1}^q\langle v_j|B_1|v_j\rangle+\ldots+
\sum_{j=1}^q\langle v_j| B_k|v_j\rangle\\&\leq&\sum_{j=1}^q\langle
u_j|A_1|u_j\rangle+\ldots+
\sum_{j=1}^q\langle u_j| A_k|u_j\rangle\\&=&\sum_{j=1}^q\langle
u_j|A_1+\ldots+
A_k|u_j\rangle =\nu_1+\ldots+ \nu_q,
\end{eqnarray*}
where $\nu_1,\ldots,\nu_q$ are the top $q$ eigenvalues of
$A_1+\ldots+A_k$. This proves the lemma. 
\end{proof}

To finish the proof of the majorization in Theorem~\ref{thm:majorization} we now show that the operators
$A_\ell=\cT^\ell(\gamma^{(k-\ell)}_{\otimes^Mu} )$ indeed satisfy
the simultaneously diagonalization property, i.e., are monotone
functions of the same operator. 
On $\cH_{k-\ell}$ we have, in terms of second quantization,
$$\gamma^{(k-\ell)}_{\otimes^Mu}=
\gamma^{(k-\ell)}_{\otimes^Mu}
=\frac{M!}{(M-k+\ell)!}|\otimes^{k-\ell}u\rangle\langle\otimes^{k-\ell}u|
=\frac{M!}{(M-k+\ell)!(k-\ell)!}a^*(u)^{(k-\ell)} a(u)^{(k-\ell)}.
$$
It follows that if $v_2,\ldots,v_N$ are chosen so that they form an
orthonormal basis
of $\C^N$ together with $u$ then on $\cH_k$ we have
\begin{eqnarray*}
\lefteqn{\cT^\ell(a^*(u)^{(k-\ell)} a(u)^{(k-\ell)})=}&&\\&&\sum_{j=0}^\ell
{\ell\choose j}a^*(u)^{(k-\ell+j)} a(u)^{(k-\ell+j)}
\sum_{i_1=2}^N\ldots\sum_{i_{\ell-j}=2}^Na^*(v_{i_1})\cdots
a^*(v_{i_{\ell-j}})a(v_{i_{\ell-j}})\cdots a(v_{i_1})
\\&=&
\sum_{j=0}^\ell
\frac{\ell!(k-\ell+j)!}{j!}{\bf 1}_{a^*(u)a(u)=k-\ell+j}.
\end{eqnarray*}

This simply says that 
 
$$
\cT^\ell(a^*(u)^{(k-\ell)} a(u)^{(k-\ell)})=f_\ell(a^*(u)a(u)),
$$
where
$$
f_\ell(m)=
\begin{cases}
\frac{\ell!m!}{(m-(k-\ell))!},  &\text{if $m\geq k-l$;}\\
0. &\text{if $m<k-l$,}
\end{cases}
$$
i.e., they are increasing functions. 
Hence using Lemma~\ref{lm:majorsums} and (\ref{eq:chiribella}) we find that 
\begin{equation}
\label{eq:Wcomp}
\cW_k(|\otimes^Mu\rangle\langle\otimes^Mu|)=\sum_{\ell=0}^kC_\ell \cT^\ell(\gamma^{(k-\ell)}_{\otimes^Mu} )\succ
\sum_{\ell=0}^kC_\ell \cT^\ell(\gamma^{(k-\ell)}_\psi )=\cW_k(|\psi\rangle\langle\psi|).
\end{equation}
which shows (\ref{eq:Wmajorization}).

To show that majorization is strict for $k>0$ unless $\psi$ is a
coherent state vector $\otimes^M u$ we shall consider $0<k\leq
q(M-1)$ and do induction on $q=1,2\ldots$.  For $q=1$ we are
assuming that $0<k\leq M-1$.  If the eigenvalues of $\cW_k(|\psi\rangle\langle\psi|)$ are equal
to the eigenvalues of $\cW_k(|\psi\rangle\langle\psi|)$ we conclude from (\ref{eq:Wcomp}), in
particular, that the eigenvalues of $\cT^0(\gamma^{k}_\psi
)=\gamma^{k}_\psi $ are the same as the eigenvalues of
$\cT^0(\gamma^{k}_{\otimes^Mu})=\gamma^{k}_{\otimes^Mu}$ which is
rank 1. Hence $\gamma^{k}_\psi$ is rank one. The result now follows
from the following lemma.
\begin{lm} 
  If, for $0<k<M$, the $k$-particle reduced density matrix
  $\gamma^{(k)}_\psi$ for a $\psi\in\cH_M$ is rank one then $\psi$
  must be a coherent state vector $\otimes^Mu$ for some unit vector
  $u$.
\end{lm}
\begin{proof} Since $\psi\in \bigotimes^M\C^N$ we can think of
  $\gamma^{k}_\psi$ (up to an overall multiplicative factor) as the
  partial trace of $|\psi\rangle\langle\psi|$ over the first $M-k$
  factors. If this is rank one so is $\gamma^{M-k}_\psi$ and moreover
  $|\psi\rangle\langle\psi|$ must be a tensor product of these two
  rank one operators, i.e., there is $\psi_{M-k}\in\cH_{M-k}$ and
  $\psi_k\in \cH_k$ such that $\psi_M=\psi_{M-k}\otimes\psi_k$.  By
  taking repeated partial traces we easily see that
  $|\psi\rangle\langle\psi|=M^{-M}\bigotimes^{M}\gamma_\psi^{1}$,
  which implies that $\gamma_\psi^{1}$ must be rank one as claimed. 
\end{proof}
To do the induction step we assume that we have proved the claim for
all $k\leq(q-1)(M-1)$ we consider $k\leq q(M-1)$. In (\ref{eq:Wcomp})
all terms in the sum on the left must have the same eigenvalues as the
terms in the sum on the right. If we consider the term for
$l=(q-1)(M-1)$ we see that $\cT^{\ell}(\gamma_\psi^{k-\ell})$ has the
same eigenvalues as $\cT^{\ell}(\gamma_{\otimes^Mu}^{k-\ell})$. Since
$\gamma_{\otimes^Mu}^{k-\ell}$ is proportional to
$|\otimes^{(k-\ell)}u\rangle\langle\otimes^{(k-\ell)}u|$ it follows
from the induction assumption that $\gamma_\psi^{(k-\ell)}$ has to be
proportional to a coherent state, in particular, rank one.  Note that
$k-\ell\leq q(M-1)-(q-1)(M-1)=M-1$ and the result again follows from
the lemma above.
\end{proof}
One of the key observations in our proof was that the two maps $\cW_k$
and $\cT_k$ acting on the same pure states will output operators with
the same eigenvalues. In particular this implies that we have also
determined the minimal output entropy of the channel $\widehat\cW_k$
defined as the trace preserving normalization of the map $\cW_k$ given
in (\ref{eq:Wdef}).
\begin{cl}[Minimal Output Entropy of $\widehat\cW_k$]\label{cl:Wk} For any concave function
  $f:[0,1]\to\R$ and any density matrix $\rho$ on $\cH_M$ we have
$$
\Tr_{M+k}f(\widehat\cW^k(\rho))\geq \Tr_{k}f(\widehat\cW^k(|\otimes^Mu\rangle\langle \otimes^Mu|)).
$$
For $M\geq 2$ and $k>0$ equality holds if and only if
$\rho=|\otimes^Mu\rangle\langle \otimes^Mu|$ for some unit vector
$u\in\C^N$.
\end{cl}

\section{The Semiclassical Limit}\label{sec:semiclassics}
In this section we study the limit of the maps $\cT^k$ as $k$ tends to
infinity with the goal of proving the generalized Wehrl inequality
Theorem~\ref{thm:wehrl}.  The limit $k\to\infty$ will turn out to be a
semiclassical limit, where the limiting object is a map from
operators $\Gamma$ on $\cH_M$ to functions on a classical phase
space. As we have discussed the classical phase space is the space of
pure quantum states on the one-body space $\C^N$, i.e., the complex
projective space $\C\P^{N-1}$ of unit vectors in $\C^N$. This space
will not play a role in our analysis as we will use (\ref{eq:su-int}) and
simply work on $\{u\in\C\ :\ |u|=1\}$, i.e., a $(2N-1)$-dimensional real sphere.

Since $SU(N)$ acts irreducibly on the symmetric space $\cH_M$ (see
appendix~\ref{sec:appendix}) we have the usual coherent states
decomposition on $\cH_M$.
\begin{equation}\label{eq:decomp}
\dim(\cH_M)\int_{u\in \C\P^{N-1}}|\otimes^M u\rangle\langle \otimes^Mu|
d_{\C\P^{N-1}}u=I_{\cH_M}.
\end{equation}
That the operator on the right is proportional to the identity follows
from Schur's Lemma and the identity then follows from the fact that
both sides have the same trace.  Recall, that the measure on the sphere 
is assumed to be normalized.

Using the Berezin-Lieb inequality \cite{Be,L1}, which states that
traces of concave functions are bounded above and below by analogous
semiclassical expressions, we can compare the finite dimensional
traces to integrals. 
\begin{lm}\label{lm:BL} If $f$ is a concave function and $\Gamma$ is
  a positive semi-definite operator on $\cH_M$ we have
\begin{equation}\label{eq:BL}
\frac1{\dim\cH_{M+k}}\Tr_{\cH_{M+k}}\left[f\left(\frac{M!}{(M+k)!}
\cT^k(\Gamma)\right)\right]
\leq \int_{\{u\in\C^N\ :\ |u|=1\}}f\left(\langle \otimes^Mu|\Gamma|
\otimes^Mu\rangle\right)du.
\end{equation}
\end{lm}
\begin{proof}
  Using the decomposition (\ref{eq:decomp}) to rewrite traces, we find 
from Jensen's inequality and the concavity of $f$ that 
\begin{eqnarray*}
\lefteqn{\frac1{\dim\cH_{M+k}}\Tr_{\cH_{M+k}}\left[f\left(\frac{M!}{(M+k)!}
\cT^k(\Gamma)\right)\right]}&&\\&=&\int_{\{u\in\C^N\ :\ |u|=1\}}
\langle \otimes^{M+k}u|f\left(\frac{M!}{(M+k)!}
\cT^k(\Gamma)\right)|\otimes^{M+k}u\rangle du\\
&\leq& \int_{\{u\in\C^N\ :\ |u|=1\}}
f\left(\frac{M!}{(M+k)!}
\langle \otimes^{M+k}u|\cT^k(\Gamma)|\otimes^{M+k}u\rangle\right)du \\
&=&\int_{\{u\in\C^N\ :\ |u|=1\}}f\left(\langle \otimes^Mu|\Gamma|
\otimes^Mu\rangle\right)du.
\end{eqnarray*}
The inequality from the first to the third line is the upper
bound in the Berezin-Lieb inequalities. In the last step we used that 
$$
\langle \otimes^{M+k}u|\cT^k(\Gamma)|
\otimes^{M+k}u\rangle=\frac{(M+k)!}{M!}\langle \otimes^Mu|\Gamma|
\otimes^Mu\rangle.
$$
\end{proof}
According to Theorem~\ref{thm:main} we have if $\Tr_M\Gamma=1$ that 
\begin{equation}\label{eq:mainappl}
\Tr_{\cH_{M+k}}\left[f\left(\frac{M!}{(M+k)!}
\cT^k(\Gamma)\right)\right]\geq 
\Tr_{\cH_{M+k}}\left[f\left(\frac{M!}{(M+k)!}
\cT^k(|\otimes^M v\rangle\langle\otimes^M v|)\right)\right]
\end{equation}
We will now study the limit $k\to\infty$ of the right side above. 
We have already seen at the end of the last section that the
eigenvalues of 
$$
\cT^k(|\otimes^M
v\rangle\langle \otimes^Mv|)=M!^{-1}\cT^k(a^*(v)^{M}a(v)^{M})
$$ 
are 
$$
\frac{k!m!}{M!(m-M)!}.
$$
The multiplicity of this eigenvalue is the number of ways we can choose
$N-1$ non-negative
integers summing up to $M+k-m$, i.e., 
$$
{M+k-m+N-2\choose N-2}.
$$
Similarly the dimension of $\cH_M$ is the number of ways in which we can
choose $N$ non-negative integers to sum up to $M$, i.e., 
$$
\dim\cH_M={M+N-1\choose N-1}.
$$
We find that 
\begin{eqnarray*}
\lefteqn{\frac1{\dim\cH_{M+k}}\Tr_{\cH_{M+k}}\left[f\left(\frac{M!}{(M+k)!}
\cT^k(|\otimes^M
v\rangle\langle \otimes^Mv|)\right)\right]}&&\\&=&
{M+k+N-1\choose
N-1}^{-1}\sum_{m=M}^{M+k}f\left(\frac{M!}{(M+k)!}\frac{k!m!}{M!(m-M)!}
\right){M+k-m+N-2\choose
  N-2}
\\&=& \frac{(N-1)(M+k)!}{(M+k+N-1)!}\sum_{m=M}^{M+k} 
f\left(
  \frac{k!m!}{(M+k)!(m-M)!}\right)\frac{(M+k-m+N-2)!}{(M+k-m)!}\\
&=&(N-1)
\sum_{m=M}^{M+k} 
f\left(
\frac{k!m!}{(M+k)!(m-M)!}\right)\frac{(M+k)!(M+k-m+N-2)!}{(M+k-m)!(M+k+N-2)!
}
\\&&\qquad\qquad\qquad\times\frac1{(M+k+N-1)}.
\end{eqnarray*}
It is straightforward to check that as $k\to\infty$ this converges for continuous\footnote{Since $f$ is
  assumed to be concave it is continuous except possibly at the
  endpoints. Discontinuity at the endpoints is not a problem.} $f$ to the integral
\begin{equation}\label{eq:s-int}
(N-1)\int_0^1 f(s^M)(1-s)^{N-2}ds.
\end{equation}
Theorem~\ref{thm:wehrl} follows from Lemma~\ref{lm:BL},
(\ref{eq:mainappl}), and the above calculation if we can show that
$$
(N-1)\int_0^1 f(s^M)(1-s)^{N-2}ds=\int_{\{u\in\C^N\ :\ |u|=1\}}
f\left(\left|\langle \otimes^Mu|\otimes^Mv\rangle\right|^2\right)du.
$$
This is simple. We choose $v=(1,0,\ldots,0)$ then 
$\left|\langle \otimes^Mu|\otimes^Mv\rangle\right|^2=|u_1|^{2M}$, where $u_1$ 
is the first coordinate of $u\in\C^N$. The $(2N-2)$-dimensional measure of the set 
$\{u\in \C^N\ :\ |u|=1, |u_1|=\cos(t)\}$ is proportional to $\cos(t)\sin^{2N-3}(t)$.
Hence
\begin{eqnarray*}
\int_{\{u\in\C^N\ :\ |u|=1\}}
f\left(\left|\langle \otimes^Mu|\otimes^Mv\rangle\right|^2\right)du&=&
\int_{\{u\in\C^N\ : \ \ |u|=1\}\}} f(|u_1|^{2M})du\\&=&2
(N-1)\int_0^{\pi/2} f(\cos^{2M}(t))\cos(t)\sin^{2N-3}(t)dt,
\end{eqnarray*}
which is equal to the integral (\ref{eq:s-int}).

Since our proof of the generalized Wehrl inequality in
Theorem~\ref{thm:wehrl} is based on a limiting argument it does not
establish that coherent states are the only minimizers. In the case of
Glauber states this was proved by Carlen in \cite{Ca}. In contrast,
for finite $k$, the uniqueness is established in
Theorem~\ref{thm:majorization} .

\appendix\section{Irreducibility of the Symmetric Representation}\label{sec:appendix}

We explain here that the Hilbert space $\cH_M$ of totally symmetric
products (see (\ref{eq:pm})) gives an irreducible unitary representation of
$SU(N)$. That it is a unitary representation is clear. If $V$ is an
invariant subspace for the action of the group it is also invariant for the
action of the representation of the Lie-algebra. The Lie-algebra of \su is
the real vector space of traceless anti-Hermitian $N\times N$ matrices. 
If $X$ is any such matrix it is represented on $\cH_M$ by
$$\pi(X)=X_1+\ldots+X_M,$$
where $X_i$ is $X$ acting on the $i$-th tensor factor. Such matrices
hence leave $V$ invariant.  Taking complex linear combinations of
matrices of the form $\pi(X)$ and the identity $I$ we find that
$X_1+\ldots+X_M$ leaves $V$ invariant for all $N\times N$ matrices $X$
(not only those that are traceless Hermitean).  In particular, $V$ is left invariant by
$E_{ij}=X_1+\ldots X_M$ with $X=|e_i\rangle\langle e_j|$ for $e_i, e_j$ elements in an
orthonormal basis $e_1,\ldots,e_N$ in $\C^N$. The
matrices $E_{ii}$, $i=1,\ldots, N$ can be simultaneously diagonalized
on $V$, which hence must contain at least one common eigenvector,
i.e., a vector of the form
\begin{equation}\label{pee}
P_M
(\otimes^{n_1}e_{1})\otimes(\otimes^{n_2}e_{2})\cdots(\otimes^{n_
{ N-1 } } e_ { N- 1} )\otimes (\otimes^{n_N}e_{N}),
\end{equation}
where $n_1+\cdots+n_N=M$. Multiple applications of the matrices
$E_{ij}$ on  one vector of the form (\ref{pee}) can yield
(multiples) of all vectors of the form (\ref{pee}). Consequently, they must 
all belong to $V$, and thus $V$ is all of $\cH_M$. \hfill \qedsymbol



\end{document}